\documentclass[letterpaper, 10pt, conference]{Classes/ieeeconf}
\IEEEoverridecommandlockouts

\usepackage{Packages/custom}

\begin{document}
    \title{\LARGE \bf Performance-Guaranteed Reference Tracking With Power Directionality Constraints: Application to Controlled Stochastic Watersheds
}

\author{Jonathan Shell$^1$, Sepehr Moalemi$^2$, Branko Kerkez$^1$ and Jeff Scruggs$^1$%
\thanks{$^{1}$ J. Shell, J. Scruggs and B. Kerkez are with the Department of Civil and Environmental Engineering, and $^{2}$ S. Moalemi is with the Department of Electrical and Computer Engineering, both at the University of Michigan, Ann Arbor, MI 48109, USA.
{\tt\footnotesize \{jonshell, moalemi, bkerkez, jscruggs\}@umich.edu}.
The authors were supported by NSF grants 2206018 and 2426817.
}}

\maketitle
    \begin{abstract}
    Modern stormwater infrastructure faces increased demands that require a corresponding increase in capacity. Traditionally, these demands have been met by constructing new infrastructure assets, which is a costly endeavor. More recently, many system operators have achieved great success in employing feedback control techniques to improve system performance. However, the resulting closed-loop system exhibits power directionality constraints that introduce nonlinear constraints in feedback synthesis. In this work, we develop a stochastic control synthesis procedure with provable performance bounds on mean-square reference tracking for a general class of problems in which power directionality constraints arise. The proposed method is then applied to a flood mitigation example using a numerical model of a real-world smart water system. The key result is an extension of the performance-guaranteed control (PGC) framework, which was originally designed for disturbance rejection, to accommodate reference tracking control objectives.
\end{abstract}
    \section{Introduction}
Modern water system infrastructure assets face increased demands that in many cases exceed their original design specifications. These increased demands are the result of population growth and the corresponding increase in impermeable land cover, and many system operators now face the prospect of costly construction projects to replace the insufficient designs. This challenge is the result of the traditional design process, in which no feedback control is implemented, and the system is operated in an open-loop environment. Recently, water utilities have considered retrofitting stormwater drainage systems with sensing and actuation to enable feedback control. Feedback control of smart water systems can be used to reduce the risk of flooding or droughts, without the need for costly new infrastructure development \cite{cembrano_optimal_2004,wong_real-time_2018,pyke_assessment_2011,garcia_modeling_2015, evans_real-time_2011,maiolo_use_2020, sun_real-time_2020,vanrolleghem_modelling_2005}. This paradigm shift towards so-called ``smart'' water systems has been enabled by the decreasing costs of sensing and actuation. Feedback control of these systems is achieved via controllable elements to direct the flow of water as well as sensors to measure important properties, such as depth or flow. Often, the cost to implement feedback control to meet rising demand is orders of magnitude less than that required for infrastructure expansion projects.

For the purposes of control design, a smart watershed can be modeled as a network of connected reservoirs and conduits, as shown in Figure~\ref{fig:model}. Basins where water collects in the system are modeled as reservoirs with the controllable flows between reservoirs serving as control inputs to the system. Depending on the network in question, some or all of these elements may be open to the atmosphere. For example, in Figure~\ref{fig:model} the reservoirs are open to the atmosphere but connected with closed pipes, each of which has a controllable element. The rainfall and surface runoff into each basin can be thought of as exogenous stochastic disturbances, while uncontrolled inter-reservoir flows can be captured in the system dynamics.
\begin{figure}[t]
    \centering
    \vspace{8pt}
    \includegraphics[width=\columnwidth]{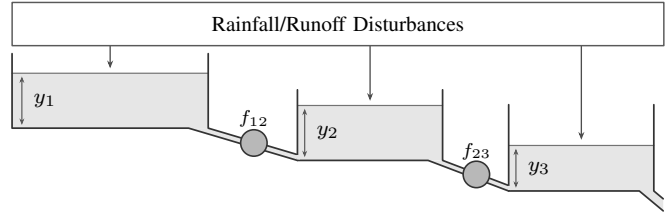}
    \vspace{-15pt}
    \caption{Connected reservoir model of a watershed.}
    \vspace{-10pt}
    \label{fig:model}
\end{figure}

The typical control variable for a smart water system is the volumetric flow through its actuated devices, depicted in Figure~\ref{fig:model} as the flow variables $f_{ij}$. Actuation technologies in smart water systems can be classified as active or passive technologies. In the case of active technologies, the typical actuation scheme involves pumping water between reservoirs, while in the passive regime controllable orifices like gates and valves are typically used. In this work, we focus on the control of systems with passive control devices. Passive technologies require much less power to operate, and so are often better suited for the task of controlling smart water systems, which may have large actuation networks or be located in remote areas without reliable access to the electricity grid.

Despite its advantages, the use of passive actuation technologies introduces additional difficulties to the control design process. One difficulty that arises is a sector constraint on the control input. For a controllable orifice, the flow through the device is driven by the pressure differential across the device, and the resulting sector constraint constitutes a power directionality constraint (PDC) on the closed-loop system. PDCs naturally arise in many passive systems, and control of systems subject to PDCs has given rise to a rich body of work \cite{hrovat_survey_1997, karnopp_active_1995, housner_structural_1997,casciati_active_2012,cha_comparative_2013}. Actuation technologies that exhibit PDCs are attractive for their inherent inability to destabilize a passive plant, providing a guarantee of stability and some measure of robustness to the closed-loop system. However, control design in the presence of PDCs is challenging due to the inherent nonlinearity of the problem. 

This paper focuses on the development of an optimal control strategy for plants which exhibit PDCs. Specifically, the aim is to synthesize a feasible feedback law that seeks to minimize the stationary mean-square deviation of the system output $y_k$ from a desired set point $y^\star$, i.e.,
\begin{equation}\label{eq:obj-fun}
    J = \E{\lVert y_k-y^\star\rVert_2^2},
\end{equation}
where $\mathcal{E}\{\cdot\}$ denotes the stationary expectation. The approach developed here is a direct extension of the performance-guaranteed control (PGC) technique \cite{scruggs_nonlinear_2006,ligeikis_discrete-time_2022,ligeikis_lqg-inspired_2022}. PGC is a control law synthesis technique for stochastically excited systems that uses online optimal control techniques to find a control input with provable mean-square performance bounds and has been particularly effective for disturbance rejection problems involving PDCs for a variety of engineering applications \cite{scruggs_nonlinear_2006, ligeikis_lqg-inspired_2022, ligeikis_discrete-time_2022}. 

Although PGC has been widely investigated for disturbance rejection across many application domains, a framework for the reference tracking problem has yet to be developed. Additionally, all past work on PGC development has assumed a standard (i.e., Gaussian) model for the disturbance. In this work, we develop a PGC algorithm for a new class of problems in which the standard noise assumptions are relaxed and it is only assumed that the first and second moments of the disturbance are known. In addition, we develop PGC for the reference tracking problem, which introduces additional complexity in achieving performance guarantees analogous to those of disturbance rejection PGC.
\begin{figure}[t]
    \centering
    \vspace{4pt}
    \includegraphics[width=\columnwidth]{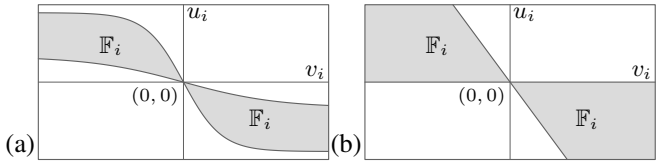}
    \vspace{-10pt}
    \caption{(a) General PDC and (b) sector-bounded constraint.}
    \vspace{-10pt}
    \label{pdc_general}
\end{figure}

Our approach is to develop the theory for stochastic reference-tracking PGC in a generic framework, and then to specialize this framework to the specific target application of urban watershed control. In Section~\ref{sec:modeling}, we first discuss the general assumptions that must hold for the subsequent theory to be valid, and then further specialize these modeling assumptions to the target application. In Section~\ref{control_section}, we present the proposed stochastic control synthesis methodology for the generic assumptions. The approach constitutes a two-stage process. In the first stage, we design a feasible linear time-invariant (LTI) static feedback controller, and optimize its performance. In the second stage, we design a nonlinear full-state controller that is guaranteed to improve upon the performance attained in the first stage. Section~\ref{sec:example} illustrates the application of this general methodology to an urban watershed control problem and validates the controller using a high-fidelity nonlinear hydrologic simulation. Finally, Section~\ref{sec:conclusions} provides brief concluding remarks.
    \section{System Modeling}\label{sec:modeling}

\subsection{General Assumptions}
We assume a discrete-time LTI system model with sampling period $\tau$ and state-space representation
\begin{equation}\label{eq:plant}
\mathcal{P}:\left\{
\begin{aligned}
    x_{k+1} &= A x_k + B_u u_k + B_w w_k,\\
    v_k     &= C_v x_k + D_u u_k + D_w w_k,\\
    y_k     &= C_y x_k,
\end{aligned}
\right.
\end{equation}
where $x_k\in \mathbb{R}^n$ is the system state, $y_k \in \mathbb{R}^{n_y}$ is the vector of outputs used for the reference tracking objective, and $w_k\in\mathbb{R}^{n_w}$ is an exogenous stochastic disturbance vector. We assume that $w_k$ constitutes an independent and identically distributed stochastic sequence, with mean $\mu_w\in\mathbb{R}^{n_w}$ and covariance matrix $S_w \in \mathbb{R}^{n_w\times n_w}$, with $S_w=S_w^T\succ 0$. The control inputs $u_k \in \mathbb{R}^{n_p}$ are colocated with the output vector $v_k \in \mathbb{R}^{n_p}$, such that the quantity $u_{i,k} v_{i,k}$ constitutes the power injected into $\mathcal{P}$ by the $i^{\mathrm{th}}$ control input. We assume that the mapping $u \mapsto v$ is passive, and consequently that the transfer function $H_{uv}(z) \triangleq D_{u} + C_v [zI-A]^{-1} B_u$ is positive-real.

We assume that the control inputs are constrained such that for each $k \in \mathbb{Z}$ and for each scalar control input $u_{i,k}$, in order for it to be feasible, the pair $(v_{i,k},u_{i,k}) \in \mathbb{F}_i$ where $\mathbb{F}_i$ is some connected subset of quadrants II and IV, as shown in Figure~\ref{pdc_general}(a). In particular, we develop our theoretical results under the assumption that the power directionality constraint is a sector, as shown in Figure~\ref{pdc_general}(b). Under this assumption, we can relate $u_k$ to a bilinear control input matrix $\tilde{Z}_k$, as 
\begin{equation}
    u_k = -\tilde{Z}_k v_k,
    \label{u=Zv}
\end{equation}
where $\tilde{Z}_k = \operatorname{diag}\{\tilde{z}_{1,k},\dots,\tilde{z}_{n_p,k}\}$ and, without loss of generality, we assume $u_k$ and $v_k$ are normalized such that $\tilde{z}_{i,k} \in [0,1]$ (in other words, each $\tilde{z}_{i,k}$ serves to modulate the corresponding $v_{i,k}$). Note that these bounds can be written equivalently as 
\begin{equation}\label{eq:z_tw_con}
    \tilde{Z}_k^T\tilde{Z}_k - \tfrac{1}{2} \tilde{Z}_k-\tfrac{1}{2} \tilde{Z}_k^T \preceq 0.
\end{equation}
Then, we have that
\begin{align*}
    u_k &= -\tilde{Z}_k(C_v x_k + D_{u} u_k + D_{w} w_k) \\
        &= -Z_k C_v x_k -Z_k D_{w} w_k,
\end{align*}
where $Z_k \triangleq [I+\tilde{Z}_{k}D_{u}]^{-1}\tilde{Z}_k$. Using this change of variable, it follows from \eqref{eq:z_tw_con} that
\begin{equation} \label{eq:z_ineq}
    Z_k^T[I + \tfrac{1}{2}(D_{u} + D_{u}^T)]Z_k - \tfrac{1}{2} Z_k^T-\tfrac{1}{2} Z_k \preceq0.
\end{equation}

\subsection{Application to Watersheds}

In the context of a watershed application, the control inputs are the volumetric flows through the controllable valves, and the colocated outputs are the corresponding pressure differentials across the valves. The fact that the valves cannot inject energy implies that the product of each flow and pressure differential must be nonpositive at each point in continuous time. Because our model is developed in discrete time rather than continuous-time, we must approximate the continuous-time power directionality constraint by a discrete-time version. Here, we do this by assuming that each component of the control input vector $u_k$ is a volumetric flow $f_{ij}(t)$ between reservoirs $i$ and $j$, where we assume $f_{ij}(t)$ is approximately constant for $t \in [k\tau , (k+1)\tau]$. Then, taking $h_k$ to be the vector of corresponding pressure differentials across the valves at time $k$, we impose the discrete-time constraint
\begin{equation*}
    u_{i,k}\tfrac{1}{2}[h_{i,k}+h_{i,k+1}] \leq 0, \qquad i=1,\dots,n_p,\quad k\in\mathbb{Z}.
\end{equation*}
We presume $h_k$  to be linearly related to the states, i.e., $h_k = C_h x_k$, and thus obtain the colocated vector $v_k$ as in \eqref{eq:plant} with 
\begin{align*}
    C_v &= \tfrac{1}{2}C_h(I + A), &
    D_{u} &= \tfrac{1}{2} C_hB_u, &
    D_{w} &= \tfrac{1}{2} C_hB_w.
\end{align*} 

Hydrologic networks exhibit relationships between pressure and flow which are more complex than \eqref{u=Zv}, because the true feasibility domain $\mathbb{F}_i$ is not a sector. From Bernoulli's principle, assuming fully submerged orifices and steady-state flow, the flow through each valve from reservoir $i$ to $j$ can be approximated in continuous time as 
\begin{equation}\label{eq:vol_flow}
    f_{ij}(t) \approx  -\operatorname{sgn}(h_{ij}(t)) K_{ij}\sqrt{|h_{ij}(t)|},
\end{equation}
where $h_{ij}(t)$ is the hydraulic head differential between basins $i$ and $j$, and $K_{ij} = C_{d,ij}A_{ij}\sqrt{2g}$ with $C_{d,ij}$ and $A_{ij}$ being the discharge coefficient and effective cross-sectional area of each orifice, respectively \cite{wong_real-time_2018}. Clearly, there is a nonlinear relationship between the hydraulic head differential $h_{ij}(t)$ and the flow $f_{ij}(t)$, but it can be conservatively approximated by a sector condition for $|h_{ij}|$ within a specified bound. This is illustrated in Figure~\ref{fig:pdc}, which shows the true power directionality constraint and a sector approximation which is conservative for pressure differentials below a certain threshold that depends on $K_{ij}$.
\begin{figure}
    \centering
    \vspace{4pt}
    \includegraphics[width=\columnwidth]{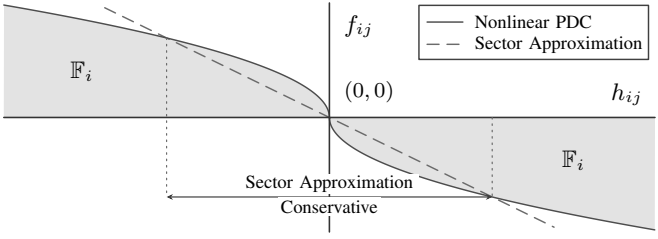}
    \caption{PDC for the watershed system and its approximate sector constraint.}
    \label{fig:pdc}
\end{figure}

Rainfall disturbances can be modeled in a variety of ways \cite{northrop_stochastic_2024, russo_rainfall_2006}. However, the critical assumption is that the disturbance is nonnegative. This precludes the use of standard (i.e., Gaussian) disturbance models that have been previously used for PGC in the literature, and is a major motivation for relaxing the noise assumptions for the current developments. A popular rainfall model used in the literature that satisfies the assumptions on $w_k$ is the Neyman--Scott Rectangular Pulse (NSRP) model \cite{cowpertwait_further_1991}, but any stochastic model satisfying the above assumptions will suffice for the proposed methodology.

Here, the output vector $y_k$ corresponds to the basin depths in the watershed system, which we aim to maintain at a vector of desired depths $y^\star$. The reference tracking problem for watershed systems is motivated by the fact that some reservoirs in the network may constitute storage basins. In this case, it is not desirable for these reservoirs to drain entirely, but overflowing must also be avoided.
    \section{Control Synthesis} \label{control_section}
The control design problem to be addressed is the stationary mean-square output reference tracking problem, i.e., the minimization of $J$ as in \eqref{eq:obj-fun}.
We propose a framework inspired by PGC for disturbance rejection problems.
In the PGC design paradigm, the controller is designed in two stages.
In the first stage, the transformed feedback gain is fixed at $Z_k=\bar{Z}$, corresponding to $u_k=-\bar{Z}(C_vx_k+D_ww_k)$, where $\bar{Z}$ is optimized for performance subject to the feasibility condition \eqref{eq:z_ineq}.
In the second stage, a full-state, nonlinear controller is designed that results in a closed-loop system with performance guarantees relative to the stage-1 linear feedback design.

\subsection{Optimal Static Feedback}
    Let $\bar{Z} \in \mathbb{R}^{n_p\times n_p}$ be some constant matrix satisfying \eqref{eq:z_ineq}, and define $\bar{x}_k$ as the value of $x_k$ resulting from $Z_\ell = \bar{Z}$ for all $\ell\in\mathbb{Z}$. 
    Then, we have that
    \begin{equation*}
        \bar{x}_{k+1} = [A-B_u\bar{Z}C_v]\bar{x}_k + [B_w-B_u\bar{Z}D_{w}]w_k.
    \end{equation*}
    In stationarity, this stochastic process has a mean 
    \begin{equation}
    \bar{\mu}_x = [I-A+B_u\bar{Z}C_v]^{-1} [B_w-B_u\bar{Z}D_{w}] \mu_w,
    \label{bar_mux}
    \end{equation}
    and its covariance matrix \(
    \bar{S}_x \triangleq \mathcal{E} \left\{ (\bar{x}_k-\bar{\mu}_x) (\bar{x}_k-\bar{\mu}_x)^T \right\}
    \) is the solution to the Lyapunov equation
    \begin{align}
    \bar{S}_x =& [A-B_u\bar{Z}C_v] \bar{S}_x [A-B_u\bar{Z}C_v]^T 
    \nonumber \\ & 
    + [B_w-B_u\bar{Z}D_{w}] S_w [B_w-B_u\bar{Z}D_{w}]^T.
    \label{bar_Sx}
    \end{align}
    The stationary performance \eqref{eq:obj-fun} attained with $Z_k = \bar{Z}$ is then readily obtained as 
    \begin{equation*}
    \bar{J} = \trace\left\{ C_y \bar{S}_x C_y^T \right\} + \| C_y \bar{\mu}_x - y^\star \|_2^2.
    \end{equation*}

    The optimal static controller $\bar{Z}^\circ$ can be found as the solution to the optimization problem
    \begin{equation} \label{eq:static_opt}
        \hspace{-9pt}\bar{Z}^\circ = \text{sol} \left\{\begin{array}{@{}rl}
            \text{minimize}&\trace \{C_y\bar{S}_x C_y^T\} +\lVert C_y\bar{\mu}_x-y^\star \rVert_2^2\hspace{-9pt}\\
            \text{over} & \bar{\mu}_x, \bar{S}_x, \bar{Z}\hspace{-9pt}\\
            \hspace{-2pt}\text{constraints} & \eqref{bar_mux},\,\eqref{bar_Sx},\,\eqref{eq:z_ineq}.\hspace{-9pt}
        \end{array} \right.
        \raisetag{2.5\baselineskip}
    \end{equation}
    The above optimization is nonconvex in $\bar{Z}$, but it only needs to be solved once offline, and any one of several commercially available nonconvex optimization solvers may be employed.
    We also note that once the equality constraints on the optimization domain are eliminated, the dimension of the domain is just $n_p$, the size of the control input vector. 

\subsection{Performance-Guaranteed Control}
    In the second stage of the optimization, we design a nonlinear feedback mapping $\mathcal{K} : x_k \mapsto Z_k$ that is guaranteed to perform no worse than the optimized LTI static feedback gain $\bar{Z}$ found in stage 1. 
    For this, we have the following theorem, which is the central contribution of this paper.

    \begin{thm} \label{main_theorem}
    Let $\bar{Z} \in \mathbb{R}^{n_p\times n_p}$ satisfy the feasibility condition \eqref{eq:z_ineq}, and let the stationary performance with $Z_k = \bar{Z}$, $\forall k$, be $\bar{J}$.
    Let $\mathcal{K} : x \mapsto Z$ be any causal feedback law for which $Z_k$ satisfies \eqref{eq:z_ineq} for all $x_k \in \mathbb{R}^n$. 
    Let $P=P^T\succeq0$ be the unique solution to the Lyapunov equation
    \begin{equation*}
    0 = [A-B_u\bar{Z}C_v]^T P [A-B_u\bar{Z}C_v] - P + C_y^T C_y.
    \end{equation*}
    Let $L$, $F$, $R$, $M$, $G$, and $W$ be defined as
    \begin{align*}
    L \triangleq& [A-B_u\bar{Z}C_v-I]^{-T} C_y^T \left[ y^\star - C_y\bar{\mu}_x \right], \\
    F \triangleq & B_u^T P [A-B_u\bar{Z}C_v], \\
    R \triangleq & B_u^T P B_u, \\
    M \triangleq & B_w-B_u\bar{Z}D_{w}, \\
    G \triangleq & B_u^T P M S_w D_{w}^T, \\
    W \triangleq & D_{w} S_w D_{w}^T,
    \end{align*}
    and define the signals
    \begin{align*}
    m_k \triangleq & x_k - \bar{\mu}_x,&
    r_k \triangleq & F m_k + B_u^T L,&
    q_k \triangleq & C_v x_k + D_{w} \mu_w.
    \end{align*}
    Then, in stationarity, the performance of controller $\mathcal{K}$ is 
    \begin{equation}\label{theorem_1_result}
    J
    = \bar{J} + \mathcal{E} \left\{ \phi(x_k,Z_k) \right\},
    \end{equation}
    where
    \begin{align*}
    \phi(x_k,Z_k) 
    = &
    q_k^T \Delta Z_k^T R \Delta Z_k q_k
    - 2 r_k^T \Delta Z_k q_k 
    \nonumber \\ & 
    + \trace\left( W \Delta Z_k^T R \Delta Z_k - 2 G^T \Delta Z_k \right),  
    \end{align*}
    and where $\Delta Z_k \triangleq Z_k - \bar{Z}$. 
    \end{thm}
    \proof
    We have that 
    \begin{align*}
    \| y_k - y^\star \|_2^2
    = x_k^T C_y^T C_y x_k - 2 (y^\star)^T C_y x_k + \| y^\star \|_2^2.
    \end{align*}
    Note that  
    \begin{align*}
    m_{k+1} &= [A-B_u\bar{Z}C_v] m_k - B_u \Delta Z_k [C_v x_k + D_{w} w_k] 
    \nonumber \\ &\qquad + M (w_k - \mu_w).
    \end{align*}
    In terms of $m_k$, the above is equivalently expressed as
    \begin{align*}
    \| y_k - y^\star \|_2^2
    &= m_k^T C_y^T C_y m_k - 2 (y^\star)^T C_y m_k 
    \nonumber \\ &\qquad 
    + 2 m_k^T C_y^T C_y \bar{\mu}_x + \| C_y \bar{\mu}_x - y^\star \|_2^2.
    \end{align*}
    Now, use the solution for $P$ to obtain the equivalent relationship
    \begin{align*}
    \| y_k - y^\star \|_2^2
    =& - m_k^T [A-B_u\bar{Z}C_v]^T P [A-B_u\bar{Z}C_v] m_k 
    \nonumber \\ & 
    + m_k^T P m_k - 2 (y^\star)^T C_y m_k 
    \nonumber \\ & 
    + 2 m_k^T C_y^T C_y \bar{\mu}_x + \| C_y \bar{\mu}_x - y^\star \|_2^2.
    \end{align*}
    Using the evolution equation for $m_k$, we have that the above is further equivalent to 
    \begin{align*}
    &\| y_k - y^\star \|_2^2
    = - m_{k+1}^T P m_{k+1} + m_k^T P m_k 
    \nonumber \\ & \quad
    - 2 m_k^T[A-B_u\bar{Z}C_v]^T P B_u \Delta Z_k [C_v x_k + D_{w} w_k ]
    \nonumber \\ & \quad
    + 2 m_k^T [A-B_u\bar{Z}C_v]^T P M[w_k-\mu_w]
    \nonumber \\ & \quad
    - 2 [C_vx_k+D_{w}w_k]^T\Delta Z_k^T B_u^T P 
    M[w_k-\mu_w]
    \nonumber \\ & \quad
    + [C_v x_k + D_{w} w_k ]^T \Delta Z_k^T B_u^T P 
    B_u \Delta Z_k [C_v x_k + D_{w} w_k ] 
    \nonumber \\ & \quad
    + [w_k-\mu_w]^T M^T P 
    M [w_k-\mu_w]
    \nonumber \\ & \quad
    - 2 (y^\star)^T C_y m_k 
    + 2 m_k^T C_y^T C_y \bar{\mu}_x + \| C_y \bar{\mu}_x - y^\star \|_2^2.
    \end{align*}
    Since $w_k$ is independent of $x_k$ and $\Delta Z_k$, it follows that the unconditional expectation of the above is
    \begin{align*}
    &\mathcal{E} \left\{ \| y_k - y^\star \|_2^2 \right\}
    = \mathcal{E} \left\{ - m_{k+1}^T P m_{k+1} + m_k^T P m_k \right\}
    \nonumber \\ & \quad
    - 2 \mathcal{E} \left\{  m_k^T[A-B_u\bar{Z}C_v]^T P B_u \Delta Z_k [C_v x_k + D_{w} \mu_w ] \right\}
    \nonumber \\ & \quad
    - 2 \mathcal{E} \left\{ \trace \left( S_w D_{w}^T \Delta Z_k^T B_u^T P M \right) \right\} + \trace\left\{ S_w M^T P M \right\} 
    \nonumber \\ & \quad
    + \mathcal{E} \left\{ [C_v x_k + D_{w} \mu_w ]^T \Delta Z_k^T R \Delta Z_k [C_v x_k + D_{w} \mu_w ] \right\}
    \nonumber \\ & \quad
    + \mathcal{E} \left\{ \trace\left( S_w D_{w}^T \Delta Z_k^T B_u^T P B_u \Delta Z_k D_{w} \right) \right\} 
    \nonumber \\ & \quad
    + 2 \mathcal{E} \left\{ m_k^T C_y^T [C_y \bar{\mu}_x-y^\star] \right\} + \| C_y \bar{\mu}_x - y^\star \|_2^2.
    \end{align*}
    Since \(\bar{J} = \trace\left\{ S_w M^TPM\right\} + \| C_y\bar{\mu}_x-y^\star \|_2^2\), it follows that
    \begin{align}
    &\mathcal{E} \left\{ \| y_k - y^\star \|_2^2 \right\}
    = \mathcal{E} \left\{ - m_{k+1}^T P m_{k+1} + m_k^T P m_k \right\}
    \nonumber \\ & \quad
    - 2 \mathcal{E} \left\{  m_k^T F^T \Delta Z_k q_k \right\}
    + \mathcal{E} \left\{ q_k^T \Delta Z_k^T R \Delta Z_k q_k \right\}
    \nonumber \\ & \quad
    - 2 \mathcal{E} \left\{ \trace \left( G^T \Delta Z_k \right) \right\}
    + \mathcal{E} \left\{ \trace\left( W \Delta Z_k^T R \Delta Z_k \right) \right\} 
    \nonumber \\ & \quad
    + 2 \mathcal{E} \left\{ m_k^T C_y^T [C_y \bar{\mu}_x-y^\star] \right\} + \bar{J},
    \label{proof_eq1}
    \end{align}
    with the definitions for $m_k$, $q_k$, $F$, $R$, $G$, and $W$ as given in the theorem.
    Next, use the definition of $L$ to get 
    \begin{align*}
    m_k^T C_y^T [C_y \bar{\mu}_x-y^\star]
    &= m_k^T L - m_{k+1}^T L + [w_k-\mu_w]^T M^T L
    \nonumber \\ & \quad \quad 
    - [C_vx_k+D_{w}w_k]^T \Delta Z_k^T B_u^T L.
    \end{align*}
    Taking the unconditional expectation and substituting into \eqref{proof_eq1}, we have that
    \begin{align*}
    \mathcal{E}\hspace{-2pt}\left\{ \| y_k - y^\star \|_2^2 \right\}
    &= \bar{J} + \mathcal{E}\hspace{-2pt}\left\{- m_{k+1}^T P m_{k+1} + m_k^T P m_k \right\} 
    \nonumber \\ & \quad
    + 2\mathcal{E}\hspace{-2pt}\left\{ m_k^T L - m_{k+1}^T L \right\} + \mathcal{E}\hspace{-2pt}\left\{ \phi(x_k,Z_k) \right\}\hspace{-1pt}.
    \end{align*}
    In stationarity, the two expectation differences involving $m_k$ and $m_{k+1}$ vanish, resulting in \eqref{theorem_1_result}.\hspace*{\fill}\QED

    Theorem~\ref{main_theorem} suggests a control strategy for formulating a full-state feedback law $\mathcal{K} : x_k \mapsto Z_k$, in which $\phi(x_k,Z_k)$ is minimized at each time step over the feasibility domain, i.e., 
    \begin{equation} \label{eq:pgc}
        Z_k = \text{sol} \left\{\begin{array}{rl}
            \text{minimize}& \phi(x_k,Z_k) \\
            \text{over} & Z_k \\
            \text{constraint} & \eqref{eq:z_ineq}.
        \end{array} \right.
    \end{equation}
    This optimization is convex and can be solved efficiently in real time. At its optimizer, $\phi(x_k,Z_k) \leqslant 0$ for all $k \in \mathbb{Z}$, since $\phi(x_k,\bar{Z}) = 0$ for all $x_k \in \mathbb{R}^n$, and $\bar{Z}$ satisfies \eqref{eq:z_ineq}.  
    Since $\phi(x_k,Z_k) \leqslant 0$ pointwise in time, its expectation is nonpositive and therefore $J \leqslant \bar{J}$.
    As such, choosing $\bar{Z} = \bar{Z}^\circ$, the solution to the optimization problem \eqref{eq:static_opt}, we have that the controller \eqref{eq:pgc} is guaranteed to do no worse than the optimized LTI static controller.
    \begin{rem}
        The fact that the nonlinear controller \eqref{eq:pgc} provides a provable bound on mean-square stochastic performance, and that this bound is the best performance achievable with LTI static control, is the central advantage of the technique we propose here.
        However, it is worth noting that the actual \emph{margin} of improvement in $J$ beyond $\bar{J}$ is in general not solvable in closed form. 
        Indeed, no claim is made that the proposed controller is truly optimal over the domain of all nonlinear causal controllers.
        Rather, the two salient advantages of the proposed feedback synthesis approach are the simplicity of the feedback law and the meaningful bound it gives on stochastic performance.
    \end{rem}
    \begin{rem}
        In the context of the prior literature on PGC, the novelty of Theorem~\ref{main_theorem} is the fact that it provides stochastic performance guarantees on a mean-square \emph{constant-reference tracking} objective, whereas all prior work on PGC has focused on disturbance rejection.
        Of course, for unconstrained control design, this is a superficial distinction because constant-reference tracking and disturbance rejection can be made equivalent to one another by shifting the origin of the state space.
        It is only because of the presence of power directionality constraints that the distinction is substantive.  
    \end{rem}
    \begin{rem}
        Although not necessary for the target application, it is straightforward to generalize Theorem~\ref{main_theorem} to accommodate definitions of performance outputs $y_k$ that involve $u_k$ and $w_k$ directly, i.e., $y_k = C_y x_k + D_{uy} u_k + D_{wy} w_k$. 
        We have specialized the theorem to the case with $D_{uy} = D_{wy} = 0$ merely to reduce the derivations required in the proof. 
    \end{rem}
    \section{Application to Watershed Control} \label{sec:example}
\begin{figure}
    \centering
    \vspace{5pt}
    \includegraphics[width=\columnwidth]{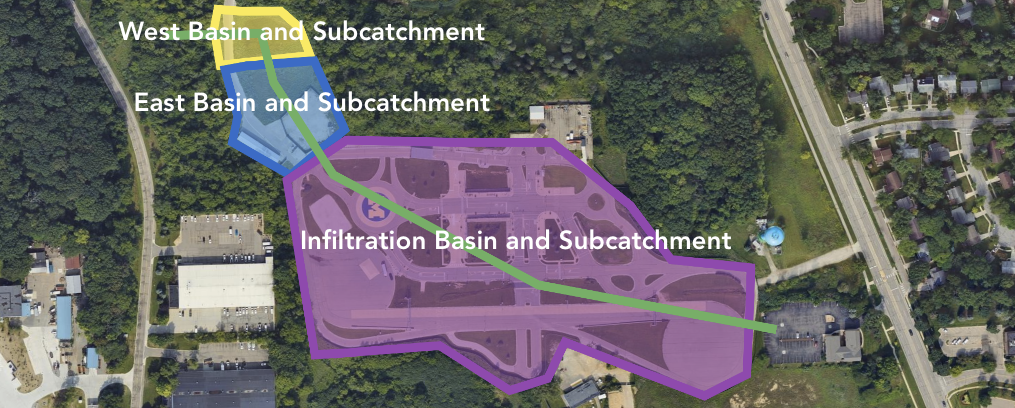}
    \vspace{-10pt}
    \caption{Satellite image of the Mcity Test Facility. The infiltration basin (pink) drains into the East Basin (blue), which drains into the West Basin (yellow), which drains into the outfall. The two controllable valves are located at the East-West and West-outfall junctions.}
    \label{fig:watershed}
    \vspace{-5pt}
\end{figure}
The watershed example considered in this section is located at the Mcity Test Facility at the University of Michigan and is shown in Figure~\ref{fig:watershed} \cite{um_mcity}. This site comprises three catchment areas and three basins with controllable valves between the East and West Basins, and between the West Basin and the outfall.

\subsection{System Identification}
    To obtain useful input-output data, a high-fidelity Stormwater Management Model (SWMM) of the watershed system was developed and used for simulations via \texttt{PySWMM}~\cite{McDonnell2020}. The NSRP model was used to generate rainfall data, which was treated as an exogenous disturbance input in the identification process. Moreover, we assumed that rainfall was spatially uniform across the watershed, and thus the same rainfall signal was applied to all three catchment areas in the model. Using \texttt{PySWMM}, time series of basin depths, junction heads, realized valve flows, and rainfall were recorded using the sampling period \(\tau=3~\mathrm{min}\).

    Subspace system identification was performed using the N4SID method~\cite{vanoverschee1994n4sid, vangemert_nfoursid_2025}. The resulting linear state-space model is of the form \eqref{eq:plant}, where $x_k \in \mathbb{R}^6$ denotes the latent state, $u_k \in \mathbb{R}^2$ is the vector of realized flows through the two controllable valves, $w_k \in \mathbb{R}$ is the scalar rainfall disturbance, $v_k \in \mathbb{R}^2$ is the colocated passive output associated with the valves, and $y_k \in \mathbb{R}^3$ is the vector of basin depths.
    After identification, a similarity transformation was applied so that the first three state coordinates correspond directly to the three basin depths.

    The initially identified input-output map from valve flow to passive output, described by $(A,B_u,C_v,D_{u})$, did not satisfy the discrete-time Kalman--Yakubovich--Popov (KYP) condition exactly, and therefore was not passive. Consequently, we computed a nearby passive approximation by solving an optimization problem, using \texttt{CVXPY}~\cite{diamond2016cvxpy}, over the passive-output matrices $\bar{C}_v$ and $\bar{D}_{u}$ while holding the identified state dynamics fixed. The differences between the identified and approximated system matrices are given by
    \begin{align*}
        \lVert C_v-\bar{C}_v\rVert_{\mathrm{F}} &= 0.171071, &
        \lVert D_{u}-\bar{D}_{u}\rVert_{\mathrm{F}} &= 0.004230.
    \end{align*}

\subsection{Controller Design}
Using the passive approximation of the identified watershed model, we next synthesized feedback controllers for stationary reference tracking of the East and West Basin depths. Note that since there is no valve actuation at the infiltration basin outlet, the depth of this basin is not directly controllable. However, by ensuring that the East Basin does not overflow, we can indirectly ensure that the infiltration basin does not overflow either, since the infiltration basin is much larger than the East Basin it feeds into. In both the static~\eqref{eq:static_opt} and PGC~\eqref{eq:pgc} controller synthesis problems, we augmented the nominal feasibility constraints with a structural diagonality constraint on \(\tilde Z = Z(I-\bar{D}_u Z)^{-1}\), requiring \([\tilde Z]_{ij}=0\) for all \(i\neq j\).

\subsection{Results}
\begin{figure}[!t]
    \centering
    \vspace{4pt}
    \includegraphics[width=1\columnwidth]{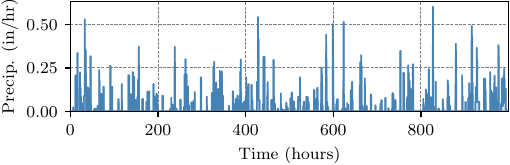}
    \vspace{-15pt}
    \caption{NSRP rainfall used for the closed-loop simulations.}
    \label{fig:rain_plot}
    \vspace{-5pt}
\end{figure}
\begin{figure}[!t]
    \vspace{4pt}
    \centering
    \includegraphics[width=1\columnwidth]{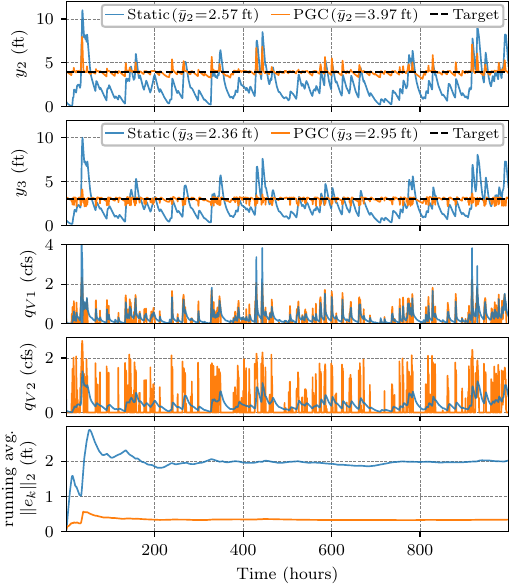}
    \vspace{-15pt}
    \caption{Closed-loop SWMM response under the optimal static feedback (blue) and PGC (orange): basin depths \(y_2\) and \(y_3\) against their targets (dashed), valve flows \(q_{V1}\) and \(q_{V2}\), and the running average of \(\lVert e_k\rVert_2\), where \(e_k=[y_{2, k} - y_2^\star,\, y_{3, k} - y_3^\star]^{T}\).}
    \label{fig:swmm_closed_loop_response}
    \vspace{-5pt}
\end{figure}
Using the NSRP rainfall data shown in Figure~\ref{fig:rain_plot}, we simulated the closed-loop response of the high-fidelity SWMM model with the static and PGC controllers, where the tracking reference was set to the vector of desired basin depths $y_{2}^\star = 4$ ft and $y_{3}^\star = 3$ ft for the East and West Basins, respectively. The closed-loop responses of the SWMM model with the static and PGC control laws as well as the tracking error are shown in Figures~\ref{fig:swmm_closed_loop_response}, where \(y_2\) and \(y_3\) correspond to the East and West Basin depths, respectively. Moreover, \(q_{V1}\) and \(q_{V2}\) correspond to the flows through the East-West and West-outfall valves, respectively. The PGC control law resulted in an 83.6\% reduction in the running-average tracking-error norm over the optimal static feedback (0.33 ft vs. 2.02 ft). 
    \section{Conclusion} \label{sec:conclusions}
Modern water infrastructure systems face increased demands, and many water utilities are considering feedback control to address them. A critical challenge in the field is the mitigation of flood and drought risk, which can be framed as a reference tracking problem. One challenge that arises in control design for smart watershed systems is the presence of power directionality constraints that result from passive actuation technologies like controllable valves. This work developed a performance-guaranteed control framework for reference tracking in systems that exhibit PDCs. The resulting closed-loop system has provable performance guarantees relative to a prescribed linear feedback control law. Future work in this area includes extending the PGC framework to a model predictive control formulation in which the control algorithm can exploit a forecast of future disturbances to achieve even better performance.
    \bibliographystyle{IEEEtran}
    \bibliography{refs.bib}

@article{cembrano_optimal_2004,
  author  = {Cembrano, G. and Quevedo, J. and Salamero, M. and Puig, V. and Figueras, J. and Mart{\'\i}, J.},
  title   = {{Optimal Control Of Urban Drainage Systems. A Case Study}},
  journal = {Control Engineering Practice},
  year    = {2004},
  volume  = {12},
  number  = {1},
  pages   = {1--9},
  doi     = {10.1016/S0967-0661(02)00280-0}
}

@article{wong_real-time_2018,
  author  = {Wong, B. P. and Kerkez, B.},
  title   = {{Real-Time Control Of Urban Headwater Catchments Through Linear Feedback: Performance, Analysis, And Site Selection}},
  journal = {Water Resources Research},
  year    = {2018},
  volume  = {54},
  number  = {10},
  pages   = {7309--7330},
  doi     = {10.1029/2018WR022657}
}

@article{pyke_assessment_2011,
  author  = {Pyke, Christopher and others},
  title   = {{Assessment Of Low Impact Development For Managing Stormwater With Changing Precipitation Due To Climate Change}},
  journal = {Landscape and Urban Planning},
  year    = {2011},
  volume  = {103},
  number  = {2},
  pages   = {166--173},
  doi     = {10.1016/j.landurbplan.2011.07.006}
}

@article{garcia_modeling_2015,
  author  = {Garc{\'\i}a, L. and Barreiro-G{\'o}mez, J. and Escobar, E. and T{\'e}llez, D. and Quijano, N. and Ocampo-Mart{\'\i}nez, C.},
  title   = {{Modeling And Real-Time Control Of Urban Drainage Systems: A Review}},
  journal = {Advances in Water Resources},
  year    = {2015},
  volume  = {85},
  pages   = {120--132},
  doi     = {10.1016/j.advwatres.2015.08.007}
}

@inproceedings{evans_real-time_2011,
  author    = {Evans, R. and others},
  title     = {{Real-Time Optimal Control Of River Basin Networks}},
  booktitle = {{IFAC} World Congress},
  series    = {{IFAC} Proceedings Volumes},
  year      = {2011},
  volume    = {44},
  number    = {1},
  pages     = {11459--11464},
  doi       = {10.3182/20110828-6-IT-1002.00451}
}

@article{maiolo_use_2020,
  author  = {Maiolo, Mario and others},
  title   = {{On The Use Of A Real-Time Control Approach For Urban Stormwater Management}},
  journal = {Water},
  year    = {2020},
  volume  = {12},
  number  = {10},
  doi     = {10.3390/w12102842}
}

@article{sun_real-time_2020,
  author  = {Sun, Congcong and Puig, Vicen{\c{c}} and Cembrano, Gabriela},
  title   = {{Real-Time Control Of Urban Water Cycle Under Cyber-Physical Systems Framework}},
  journal = {Water},
  year    = {2020},
  volume  = {12},
  number  = {2},
  doi     = {10.3390/w12020406}
}

@article{vanrolleghem_modelling_2005,
  author  = {Vanrolleghem, Peter A. and Benedetti, Lorenzo and Meirlaen, Jurgen},
  title   = {{Modelling And Real-Time Control Of The Integrated Urban Wastewater System}},
  journal = {Environ. Model. Softw.},
  year    = {2005},
  volume  = {20},
  number  = {4},
  pages   = {427--442},
  doi     = {10.1016/j.envsoft.2004.02.004}
}

@article{hrovat_survey_1997,
  author  = {Hrovat, D.},
  title   = {{Survey Of Advanced Suspension Developments And Related Optimal Control Applications}},
  journal = {Automatica},
  year    = {1997},
  volume  = {33},
  number  = {10},
  pages   = {1781--1817},
  doi     = {10.1016/S0005-1098(97)00101-5}
}

@article{karnopp_active_1995,
  author  = {Karnopp, D.},
  title   = {{Active And Semi-Active Vibration Isolation}},
  journal = {Journal of Mechanical Design},
  year    = {1995},
  volume  = {117},
  number  = {B},
  pages   = {177--185},
  doi     = {10.1115/1.2836452}
}

@article{housner_structural_1997,
  author  = {Housner, G. W. and others},
  title   = {{Structural Control: Past, Present, And Future}},
  journal = {Journal of Engineering Mechanics},
  year    = {1997},
  volume  = {123},
  number  = {9},
  pages   = {897--971},
  doi     = {10.1061/(ASCE)0733-9399(1997)123:9(897)}
}

@article{casciati_active_2012,
  author  = {Casciati, Fabio and Rodellar, Jos{\'e} and Yildirim, Umut},
  title   = {{Active And Semi-Active Control Of Structures -- Theory And Applications: A Review Of Recent Advances}},
  journal = {Journal of Intelligent Material Systems and Structures},
  year    = {2012},
  volume  = {23},
  number  = {11},
  pages   = {1181--1195},
  doi     = {10.1177/1045389X12445029}
}

@article{cha_comparative_2013,
  author  = {Cha, Young-Jin and others},
  title   = {{Comparative Studies Of Semiactive Control Strategies For MR Dampers: Pure Simulation And Real-Time Hybrid Tests}},
  journal = {Journal of Struct. Eng.},
  year    = {2013},
  volume  = {139},
  number  = {7},
  pages   = {1237--1248},
  doi     = {10.1061/(ASCE)ST.1943-541X.0000639}
}

@inproceedings{ligeikis_discrete-time_2022,
  author    = {Ligeikis, Connor and Scruggs, Jeff},
  title     = {{Discrete-Time, Performance-Guaranteed Control Of Vibratory Systems With Power Directionality Constraints}},
  booktitle = {American Control Conference},
  year      = {2022},
  pages     = {4248--4255},
  doi       = {10.23919/ACC53348.2022.9867818}
}

@inproceedings{ligeikis_lqg-inspired_2022,
  author    = {Ligeikis, Connor and Scruggs, Jeff},
  title     = {{An LQG-Inspired Framework For Self-Powered Feedback Control}},
  booktitle = {Conference on Decision and Control},
  year      = {2022},
  pages     = {4519--4526},
  doi       = {10.1109/CDC51059.2022.9993121}
}

@article{northrop_stochastic_2024,
  author  = {Northrop, Paul J.},
  title   = {{Stochastic Models Of Rainfall}},
  journal = {Annual Review of Statistics and Its Application},
  year    = {2024},
  volume  = {11},
  number  = {1},
  pages   = {51--74},
  doi     = {10.1146/annurev-statistics-040622-023838}
}

@article{russo_rainfall_2006,
  author  = {Russo, Fabio and Lombardo, Federico and Napolitano, Francesco and Gorgucci, Eugenio},
  title   = {{Rainfall Stochastic Modeling For Runoff Forecasting}},
  journal = {Physics and Chemistry of the Earth, Parts A/B/C},
  year    = {2006},
  volume  = {31},
  number  = {18},
  pages   = {1252--1261},
  doi     = {10.1016/j.pce.2006.06.002}
}

@article{cowpertwait_further_1991,
  author  = {Cowpertwait, Paul S. P.},
  title   = {{Further Developments Of The Neyman--Scott Clustered Point Process For Modeling Rainfall}},
  journal = {Water Resources Research},
  year    = {1991},
  volume  = {27},
  number  = {7},
  pages   = {1431--1438},
  doi     = {10.1029/91WR00479}
}

@electronic{um_mcity,
  author = {{Mcity}},
  title  = {{Mcity Test Facility}},
  note   = {Accessed: Apr. 1, 2026},
  url    = {https://mcity.umich.edu/what-we-do/mcity-test-facility/}
}

@article{McDonnell2020,
  author    = {McDonnell, Bryant E. and Ratliff, Katherine and Tryby, Michael E. and Wu, Jennifer Jia Xin and Mullapudi, Abhiram},
  title     = {{PySWMM: The Python Interface To Stormwater Management Model (SWMM)}},
  journal   = {Journal of Open Source Softw.},
  publisher = {The Open Journal},
  year      = {2020},
  volume    = {5},
  number    = {52},
  doi       = {10.21105/joss.02292}
}

@article{vanoverschee1994n4sid,
  author  = {Van Overschee, Peter and De Moor, Bart},
  title   = {{N4SID: Subspace Algorithms For The Identification Of Combined Deterministic-Stochastic Systems}},
  journal = {Automatica},
  year    = {1994},
  volume  = {30},
  number  = {1},
  pages   = {75--93},
  doi     = {10.1016/0005-1098(94)90230-5}
}

@electronic{vangemert_nfoursid_2025,
  author       = {van Gemert, Steven},
  title        = {{NFourSID: Implementation Of N4SID, Kalman Filtering And State-Space Models}},
  howpublished = {Python Package Index ({PyPI})},
  note         = {Version 1.0.2},
  url          = {https://pypi.org/project/nfoursid/1.0.2/},
  year         = {2025}
}

@article{diamond2016cvxpy,
  author  = {Diamond, Steven and Boyd, Stephen},
  title   = {{CVXPY: A Python-Embedded Modeling Language For Convex Optimization}},
  journal = {Journal of Machine Learning Research},
  year    = {2016},
  volume  = {17},
  number  = {83},
  pages   = {1--5}
}

@inproceedings{scruggs_nonlinear_2006,
  author={Scruggs, Jeffrey T. and Taflanidis, Alexandros A. and Iwan, Wilfred D.},
  title={{Nonlinear Stochastic Controllers for Semiactive and Regenerative Structural Control Systems, with Guaranteed Quadratic Performance Margins}},
  booktitle={Dynamic systems and controls, Symposium on Design and Analysis of Advanced Structures, tribology}, publisher={American Society of Mechanical Engineers},
  year={2006},
  pages={487--495},
  DOI={10.1115/ESDA2006-95625},
  address={New York}
}
\end{document}